\documentclass[11pt]{amsart}

\usepackage{amsmath}
\usepackage{amssymb}
\usepackage{amsfonts}
\usepackage{latexsym}
\usepackage{color}

\newtheorem{theo}{Theorem}[section]
\newtheorem{lem}[theo]{Lemma}
\newtheorem{prop}[theo]{Proposition}
\newtheorem{cor}[theo]{Corollary}
\newtheorem{lemma}[theo]{Lemma}

 \theoremstyle{definition}
\newtheorem{definition}[theo]{Definition}

\newtheorem{example}[theo]{Example}
 \theoremstyle{remark}
 \newtheorem{rem}[theo]{Remark}
 \numberwithin{equation}{section}

\newtheorem{remark}[theo]{Remark}

\newcommand{\betheo}{\begin{theo}$\!\!\!${\bf } }
\newcommand{\entheo}{\end{theo}}

\newcommand{\becor}{\begin{cor}$\!\!\!$  }
\newcommand{\encor}{\end{cor}}

\newcommand{\belem}{\begin{lem}$\!\!\!${\bf .} }
\newcommand{\enlem}{\end{lem}}

\newcommand{\beprop}{\begin{prop}$\!\!\!${\bf } }
\newcommand{\enprop}{\end{prop}}

\newcommand{\bedefi}{\begin{definition}$\!\!\!$ \rm }
\newcommand{\findefi}{ \end{definition}}

\newcommand{\beex}{\begin{example}$\!\!\!$ \rm }
\newcommand{\enex}{ \end{example}}

\newcommand{\berem}{\begin{remark}$\!\!\!$ \rm }
\newcommand{\enrem}{ \end{remark}}

\newcommand{\be}{\begin{equation}}
\newcommand{\en}{\end{equation}}

\newcommand{\bea}{\begin{eqnarray}}
\newcommand{\ena}{\end{eqnarray}}

\newcommand{\beano}{\begin{eqnarray*}}
\newcommand{\enano}{\end{eqnarray*}}

\newcommand{\bee}{\begin{enumerate}}
\newcommand{\ene}{\end{enumerate}}

\newcommand{\bei}{\begin{itemize}}
\newcommand{\eni}{\end{itemize}}

\newcommand{\betab}{\begin{tabular}}
\newcommand{\entab}{\end{tabular}}

\newcommand{\bd}{\begin{displaymath}}

\newcommand{\ad}{^{\mbox{\scriptsize $\dag$}}}
\newcommand{\das}{^{\dag {\rm\textstyle *}}}

\newcommand{\up}{\raisebox{0.7mm}{$\upharpoonright $}}%

\def\D{{\mathcal D}}

\def\F{{\mathcal F}}

\def\H{{\mathcal H}}

\def\L{{\mathcal L}}

\def\P{{\mathcal P}}

\def\J{\relax\ifmmode {\mathcal J}\else${\mathcal J}$\fi}
\def\x{\relax\ifmmode {\mbox{*}}\else*\fi}

\def\MM{{\mathfrak M}}

\newcommand{\mb}{\mathbb}

\newcommand{\mult}{{\scriptstyle \Box}}

\newcommand{\pa}{partial \mbox{*-algebra}}

\newcommand{\LD}{{\L}\ad(\D)}
\newcommand{\LDH}{{\L}\ad(\D,\H)}

\def\dag{\dagger}

\newcommand{\vp}{\varphi}

\newcommand{\ip}[2]{\left\langle {#1}\left|{#2}\right.\right\rangle}

\def\OL{\relax\ifmmode {\sf L}\else{\textsf L}\fi}
\def\OR{\relax\ifmmode {\sf R}\else{\textsf R}\fi}

\newcommand{\Id}{1\!\!1}

\begin{document}

\title[Weak commutation relations]{Weak commutation relations of unbounded operators: nonlinear extensions}

\author[F.Bagarello]{Fabio Bagarello}
\address{%
Dieetcam, Universit\`a di Palermo, 90128 Palermo, Italy} \email{bagarell@unipa.it}

\author[A.Inoue]{Atsushi Inoue}
\address{Department of Applied Mathematics, Fukuoka University, Fukuoka 814-0180, Japan}
\email{a-inoue@fukuoka-u.ac.jp}

\author[C.Trapani]{Camillo Trapani}
\address{Dipartimento di Matematica e Informatica, Universit\`a di Palermo, I-90123 Palermo, Italy}
\email{trapani@unipa.it} \maketitle

\begin{abstract}
We continue our analysis of the consequences of the commutation relation $[S,T]=\Id$, where $S$ and $T$ are two closable unbounded operators.
The {\em weak} sense of this commutator is given in terms of the inner product of the Hilbert space $\H$ where the operators act. { We also
consider what we call, adopting a physical terminology}, a {\em nonlinear} extension of the above commutation relations.
\end{abstract}

\section{Introduction}
Let $A,B$ be two closed operators with dense domains, $D(A)$ and $D(B)$, in Hilbert space $\H$. In \cite{bit_jmp2011} we have discussed  some
mathematical aspects connected to the formal commutation relation $[A,B]=\Id$. { Since, as it is well known,  $A$ and $B$ cannot be  both
bounded operators, a careful and rigorous analysis is needed. Thus, starting from the very beginning, we require that the identity $AB-BA=\Id$
holds, at least, on a dense domain $\D$ of Hilbert space $\H$. In other words, we assume that there exists a dense subspace $\D$ of $\H$ such
that
\begin{itemize}
\item[(D.1)] $\D \subset D(AB)\cap D(BA)$;
\item[(D.2)] $AB\xi-BA\xi = \xi, \quad \forall \xi \in \D$,
where, as usual, $D(AB)=\{ \xi \in D(B):\, B\xi \in D(A)\}$.
\end{itemize}
As we did in \cite{bit_jmp2011} we will suppose that
\begin{itemize}
\item[(D.3)] $\D \subset D(A^*)\cap D(B^*)$.
\end{itemize}
(D.1), (D.2) and (D.3) then imply }that the operators $S:=A\upharpoonright \D$ and $T:= B\upharpoonright D$ are elements of the partial
*-algebra $\LDH$ and satisfy the equality
$$ \ip{T\xi}{S\ad \eta}-\ip{S\xi}{T\ad \eta}= \ip{\xi}{\eta}, \quad \forall \xi, \eta \in \D.$$
We recall that $ \L\ad(\D,\H) $ denotes the set of all (closable) linear operators $X$ such that $ {D}(X) = {\D},\; {D}(X\x) \supseteq {\D}.$
The set $ \L\ad(\D,\H ) $ is a  \pa\
 with respect to the usual sum $X_1 + X_2 $,
the scalar multiplication $\lambda X$, the involution $ X \mapsto X\ad := X\x \up {\D}$ and the \emph{(weak)} partial multiplication $X_1 \mult
X_2 = {X_1}\ad\x X_2$, defined whenever $X_2$ is a weak right multiplier of $X_1$ (we shall write $X_2 \in R^{\rm w}(X_1)$ or $X_1 \in L^{\rm
w}(X_2)$), that is, whenever $ X_2 {\D} \subset {\D}({X_1}\ad\x)$ and  $ X_1\x {\D} \subset {\D}(X_2\x)$, \cite{ait_book}.

\medskip
Let $t\mapsto V(t)$, $t\geq 0$ be a semigroup of bounded operators in Hilbert space. We recall that $V$ is weakly (or, equivalently, strongly)
continuous if
$$ \lim_{t\to {t_0}}\ip{V(t)\xi}{\eta}= \ip{V(t_0)\xi}{\eta} , \quad \forall \xi, \eta \in \H.$$

A closed operator $X$ is the generator of $V(t)$ if
$$ D(X)= \left\{\xi\in \H;\exists\, \xi'\in \H: \lim_{t\to 0} \ip{\frac{V(t)-\Id}{t}\xi}{\eta}=\ip{\xi'}{\eta},\; \forall \eta \in \H \right\}$$
and
$$ X\xi = \xi', \quad \forall \xi \in D(X).$$

If $V(t)$ is a weakly continuous semigroup, then $V^*(t):=(V(t))^*$ is also a weakly continuous semigroup and if $X$ is the generator of
$V(t)$, then $X^*$ is the generator of $V^*(t)$.

An operator $X_0\in \LDH$ is the $\D$-generator of a semigroup $V(t)$ if $V(t)$ is generated by some closed extension $X$ of $X_0$ such that
$\overline{X_0}\subset X \subset X_0\das$. The latter condition ensures us that if $X_0$ is the $\D$-generator of $V(t)$, then $X_0\ad$ is the
$\D$-generator of $V^*(t)$.

\medskip
In \cite{bit_jmp2011} we gave the following definition \bedefi Let $S, T \in \LDH$. We say that
\begin{itemize}
\item[(CR.1)] the commutation relation $[S,T]=\Id_{\D}$ is satisfied (in $\LDH$) if, whenever $S\mult T$  is well-defined, $T\mult S$ is well-defined too and $S\mult T-T\mult S= \Id_{\D}$.
 \item[(CR.2)] the commutation relation $[S,T]=\Id_{\D}$ is satisfied in {\em weak sense} if
$$ \ip{T\xi}{S\ad \eta}-\ip{S\xi}{T\ad \eta}= \ip{\xi}{\eta}, \quad \forall \xi, \eta \in \D.$$
\item[(CR.3)] the commutation relation $[S,T]=\Id_{\D}$ is satisfied in {\em quasi-strong sense} if $S$ is the $\D$-generator of a weakly continuous semigroups of bounded operators $V_S(\alpha)$ and
    $$\ip{V_S(\alpha)T\xi}{\eta}-\ip{V_S(\alpha)\xi}{T\ad \eta}=\alpha\ip{V_S(\alpha)\xi}{\eta}, \quad \forall \xi, \eta \in \D; \alpha \geq 0.$$
\item[(CR.4)] the commutation relation $[S,T]=\Id_{\D}$ is satisfied in {\em strong sense} if $S$ and $T$ are $\D$-generators of weakly continuous semigroups of bounded operators $V_S(\alpha), V_T(\beta)$, respectively, satisfying the generalized Weyl commutation relation
    $$ V_S(\alpha)V_T(\beta) = e^{\alpha\beta} V_T(\beta)V_S(\alpha), \quad \forall \alpha, \beta \geq 0.$$
\end{itemize}
\findefi { As shown in \cite{bit_jmp2011}, the following implications hold

$$ \mbox{ {\rm (CR.4)} $\Rightarrow$ {\rm (CR.3)} $\Rightarrow$ {\rm(CR.2)} $\Rightarrow$ {\rm (CR.1)}.}$$
}

Our analysis was motivated by the introduction, on a more physical side, of what have been called {\em pseudo-bosons}, arising from a
particular deformation of the canonical commutation relations, see \cite{bagrep} for a recent review. Later on, one of us (FB) has extended
this notion to the so-called {\em nonlinear pseudo-bosons} , \cite{nlpb}, in which the commutation rule $[S,T]=\Id_\D$ does not hold, in
general, in any of the above meanings.  Section 4 is dedicated to a mathematical treatment of this extension,  while Sections 2 and 3 contain
more results on the {\em linear} case.

\section{Some consequences of (CR.3)}

Assume that $S,T$ satisfy the commutation relation $[S,T]=\Id_{\D}$ in quasi-strong sense; i.e.
 $$\ip{V_S(\alpha)T\xi}{\eta}-\ip{V_S(\alpha)\xi}{T\ad \eta}=\alpha\ip{V_S(\alpha)\xi}{\eta}, \quad \forall \xi, \eta \in \D; \alpha \geq 0.$$

 If we take $\xi=\eta$ and apply the Cauchy-Schwarz inequality, we get, for every $z\in {\mb C}$ and $\alpha\geq 0$,
\begin{equation} \label{UP}\alpha |\ip{V_S(\alpha)\xi}{\xi}|\leq 2 \max\{\|(T-z)\xi\|, \|(T\ad-\overline{z}) \xi\|\} \max\{\|V_S(\alpha)\xi\|, \|V_S(\alpha)^*\xi\| \}.\end{equation}

As an immediate consequence of \eqref{UP}, we get
\begin{prop} Let $S,T$ satisfy the commutation relation $[S,T]=\Id_{\D}$ in quasi-strong sense. Assume that $T=T\ad$. Then $\sigma_p(T)=\emptyset$.
\end{prop}
\begin{proof} If $\lambda$ is an eigenvalue of $T$, then the right hand side of \eqref{UP} vanishes for $z=\overline{z}=\lambda$ and $\xi=\xi_0$ a corresponding eigenvector. Hence,  $|\ip{V_S(\alpha)\xi_0}{\xi_0}|=0$. Taking the limit for $\alpha\to 0$, one gets $\|\xi_0\|^2=0$. This is a contradiction.

\end{proof}

 As we know, there exist $M>0$ and $\omega\in {\mb R}$ such that $\|V_S(\alpha)\|\leq Me^{\omega \alpha}$, for $\alpha\geq 0.$

 Let us assume that $V_S$ is uniformly bounded, i.e. $\|V_S(\alpha)\|\leq M$ for every $\alpha\geq 0$, like it happens when $V_S$ is a semigroup of isometries or a semigroup of contractions. Then, by \eqref{UP} it follows that
\begin{equation}\label{limit} \lim_{\alpha\to \infty} |\ip{V_S(\alpha)\xi}{\xi}| =0, \quad \forall \xi \in \D.\end{equation}

 \begin{lemma} \label{lemma_limit} Assume that $V_S$ is uniformly bounded. Then
 $$ \lim_{\alpha\to \infty} |\ip{V_S(\alpha)\xi}{\xi}| =0, \quad \forall \xi \in \H.$$
 \end{lemma}
\begin{proof} Let $\xi \in \H$ and  $\{\xi_n\}$ a sequence in $\D$ converging to $\xi$. We have:
\begin{align*}
|\ip{V_S(\alpha)\xi}{\xi}& - \ip{V_S(\alpha)\xi_n}{\xi_n}|\\
& = |\ip{V_S(\alpha)\xi}{\xi} - \ip{V_S(\alpha)\xi}{\xi_n} +\ip{V_S(\alpha)\xi}{\xi_n}- \ip{V_S(\alpha)\xi_n}{\xi_n}|\\
&\leq \|V_S(\alpha)\xi\|\|\xi-\xi_n\| + \|\xi-\xi_n\| \|V_S(\alpha)^*\xi_n\|\\
&\leq M (\|\xi\|+\|\xi_n\|)\|\xi-\xi_n\|.
\end{align*}
Hence
$$|\ip{V_S(\alpha)\xi}{\xi}|\leq |\ip{V_S(\alpha)\xi_n}{\xi_n}| + M (\|\xi\|+\|\xi_n\|)\|\xi-\xi_n\|.$$
Thus, by \eqref{limit}, we get
$$ \limsup_{\alpha \to \infty}|\ip{V_S(\alpha)\xi}{\xi}|\leq M (\|\xi\|+\|\xi_n\|)\|\xi-\xi_n\|, \quad \forall n \in {\mb N}.$$
This clearly implies that
$$\lim_{\alpha\to \infty} |\ip{V_S(\alpha)\xi}{\xi}| =0.$$
\end{proof}

\begin{theo} Assume that $V_S$ is a semigroup of contractions (i.e., \linebreak \mbox{$\|V_S(\alpha)\|\leq 1$}), for every $\alpha\geq 0$. Then every eigenvalue of the generator $X\supset S$ of $V_S$ has negative real part.
\end{theo}
\begin{proof} Assume that $\lambda \in {\mb C}$ is an eigenvalue of $X$. Then,  there exists $\xi \in D(X)\setminus\{0\}$ such that $X\xi= \lambda \xi$.
The Hille-Yosida theorem then implies that
$$V_S(\alpha)\xi = \lim_{\epsilon \to 0} e^{\alpha X(I-\epsilon X)^{-1}}\xi, \quad \forall \alpha \geq 0.$$
An easy computation shows that $e^{\alpha X(I-\epsilon X)^{-1}}\xi = e^{\alpha \lambda(1-\epsilon \lambda)^{-1}}\xi\to e^{\alpha \lambda}\xi$
as $\epsilon \to 0$. By Lemma \ref{lemma_limit} we conclude that $\Re(\lambda)<0$.

\end{proof}
As a special case we obtain a result already proved by Miyamoto (under additional conditions), \cite{miyamoto}.
\begin{cor} Assume that the generator $X$ of $V_S$ has the form $X=iH$ where $H$ is a self-adjoint operator. Then $\sigma_p (H)=\emptyset$.
\end{cor}

\section{Weyl extensions}
\bedefi Let $S,T$ be symmetric operators of $\LDH$. We say that $\{S,T\}$ satisfy the {\em weak Weyl commutation relation} if there exists a
self-adjoint extension $H$ of $S$ such that
\begin{itemize}
\item[({\sf ww}$_1$)]$D(\overline{T})\subset D(H)$;
\item[({\sf ww}$_2$)]$\ip{e^{-itH}\xi}{T\eta}=\ip{(T+t)\xi}{e^{itH}\eta}, \quad \forall \xi, \eta \in \D; t\in {\mb R}$.
\end{itemize}
Then $H$ is called the weak Weyl extension of $S$ (with respect to $T$). \findefi

\berem We note that we do not assume that $e^{itH}D(\overline{T}) \subset D(\overline{T}).$ \enrem \berem From ({\sf ww}$_2$) it follows that
$e^{-itH}\xi\in D(T^*)$, for every $\xi \in \D$, $t \in {\mb R}$, and
$$ T^* e^{-itH}\xi = e^{-itH}(T+t)\xi, \quad \forall \xi \in \D, t \in {\mb R}.$$

\enrem
\begin{prop} \label{prop_2.3} Let $\{S,T\}$ satisfy the weak Weyl commutation relation and let $H$ be the weak Weyl extension of $S$. The following statements hold.
\begin{itemize}
\item[(i)] Suppose that $T$ is essentially self-adjoint. Then $\{H, \overline{T}\}$ satisfy the the Weyl commutation relation, that is,
\begin{equation}\label{WCR} e^{itH}e^{-is\overline{T}} = e^{-its}e^{-is\overline{T}}e^{itH}, \quad \forall s,t \in {\mb R} \end{equation}
\item[(ii)] If $H$ is semibounded, then $T$ is not essentially self-adjoint.
\end{itemize}
\end{prop}
\begin{proof} (i): By ({\sf ww}$_1$) and ({\sf ww}$_2$) it follows that
$$\ip{e^{-itH}\xi}{\overline{T}\eta}=\ip{(\overline{T}+t)\xi}{e^{itH}\eta}, \quad \forall \xi, \eta \in D(\overline{T}); t\in {\mb R}.$$
Then, by the functional calculus, we get
$$\ip{e^{-itH}\xi}{e^{-is\overline{T}}\eta}=\ip{e^{-is(\overline{T}+t)}\xi}{e^{itH}\eta}, \quad \forall \xi, \eta \in \D; \forall s,t \in {\mb R}.$$
This, in turn, easily implies \eqref{WCR}.

(ii): Suppose that $T$ is essentially self-adjoint. Similarly to the proof of \cite[Theorem 2.7]{arai}, we have, by (i),
$$\ip{e^{-is\overline{T}}\xi}{H\eta}=\ip{H\xi}{e^{is\overline{T}}\eta}- s\ip{e^{-is\overline{T}}\xi}{\eta}, \quad \forall \xi, \eta \in D(H); s,t\in {\mb R}.$$
Put $\eta= e^{-is\overline{T}}\xi$, $\xi \in \D$, $\|\xi\|=1$. Then $\eta \in D(\overline{T})\subset D(H)$, by ({\sf ww}$_1$), and
$$ \ip{e^{-is\overline{T}}\xi}{He^{-is\overline{T}}\xi}=\ip{H\xi}{\xi}- s.$$

Hence,
$$\sup_{s,t \in {\mb R}} \ip{e^{-is\overline{T}}\xi}{He^{-is\overline{T}}\xi} = +\infty,$$
$$\inf_{s,t \in {\mb R}} \ip{e^{-is\overline{T}}\xi}{He^{-is\overline{T}}\xi} = -\infty.$$
These equalities contradict the semiboundedness of $H$.

\end{proof}
\berem A natural question is the following: When does there exist a semibounded Weyl extension of $S$? If $S$ is a semibounded symmetric
operator with finite deficiency indices, then  any self-adjoint extension of $S$ is bounded below, \cite[Proposition X.3]{reedsimon}. \enrem

\medskip
We now investigate the spectrum of $T$.

\begin{lemma}\label{lemma_4} Suppose that $H$ is bounded below. Then, for every $\beta >0$,
\begin{align*}
& e^{-\beta H}D(\overline{T})\subset D(T^*), \\
& T^* e^{-\beta H}\xi = e^{-\beta H}(\overline{T}-\beta i) \xi, \quad \forall \xi \in D(\overline{T}).
\end{align*}
\end{lemma}
\begin{proof} This is proved similarly to \cite[Theorem 6.2]{arai}.
\end{proof}

\begin{theo} Suppose that $H$ is semibounded and
\begin{equation}\label{domain} D^\infty(T^*):= \bigcap_{n \in {\mb N}}D(T^{*n}) \subset D(\overline{T}).
\end{equation}
Then $\sigma(T)={\mb C}$.
\end{theo}

\begin{proof} By \cite[Theorem X.1]{reedsimon} and (ii) of Proposition \ref{prop_2.3}, $\sigma(T)\, (=\sigma(\overline{T})$ is one of the following sets

(a) ${\mb C}$;

(b) $\overline{\Pi}_+$, the closure of the upper half-plane $\Pi_+=\{z \in {\mb C}; \Im z >0\};$

(c) $\overline{\Pi}_-$, the closure of the lower half-plane $\Pi_-=\{z \in {\mb C}; \Im z <0\}.$

Suppose that $\sigma(T)=\overline{\Pi}_-$. For every $z \in {\mb C}\setminus {\mb R}$, we have
$$ \H= {\rm Ker}(T^*-\overline{z}) \oplus R(\overline{T}-z).$$
Since $i\in \Pi_+ =\rho(\overline{T})$, we get ${\rm Ker}(T^*+i)=\{0\}$. But $\overline{T}$ is not self-adjoint (Proposition \ref{prop_2.3}),
hence  ${\rm Ker}(T^*-i)\neq\{0\}$; i.e. there exists a nonzero $\eta \in D(T^*)$ such that $T^*\eta=i\eta$. Then, $\eta \in
\D^\infty(T^*)\subset D(\overline{T})$, by \eqref{domain}. By Lemma \ref{lemma_4},
$$ T^*e^{-\beta H} \eta = e^{-\beta H}(\overline{T}-\beta i)\eta = (1-\beta)i e^{-\beta H}\eta.$$
This implies that $\gamma:= (1-\beta)i \in \sigma_p(T^*)$. Since $\H= {\rm Ker}(T^*-\gamma) \oplus R(\overline{T}-\overline{\gamma}),$ we
obtain that $R(\overline{T}-\overline{\gamma})\neq \H$. Thus $ \overline{\gamma} =(\beta-1)i \in \sigma(\overline{T})=\sigma(T)$. Now, if we
take $\beta>1$, by the assumption, $(\beta-1)i\in \Pi_+=\rho(\overline{T}) $. This is a contradiction. Therefore $\sigma(T)\neq
\overline{\Pi}_-$. In very similar way one can prove that $\sigma(T)\neq \overline{\Pi}_+$. In conclusion, $\sigma(T)={\mb C}$.

\end{proof}
\berem When does the inclusion $D^\infty(T^*)\subset D(\overline{T})$ hold? Let us consider the partial O*-algebra $\MM_w(T)$ generated by $T$
described in \cite[Section 2.6]{ait_book}. If $\MM_w(T)$ is essentially self-adjoint, then $D^\infty(T^*)\subset D(\overline{T})$. If $T\in
\LD$, then $\MM_w(T)$ is nothing but the O*-algebra $\P(T)$ of all polynomials in $T$; in this case if $\MM_w(T)$ is essentially self-adjoint,
then $T$ is essentially self-adjoint. But we know \cite[Example 2.6.28]{ait_book} that, in the case of partial O*-algebras, the essential
self-adjointness of $\MM_w(T)$ does not imply the essential self-adjointness of $T$. \enrem

\section{A nonlinear extension}

In this section we will consider a generalization of condition (CR2) introduced in Section 1, to what could be called, borrowing a physical
terminology adopted first in \cite{nlpb}, a {\em nonlinear } situation. We start  considering two biorthogonal (Schauder) bases of the Hilbert
space $\H$, both contained in $\D$, $\F_\varphi=\{\varphi_n\in\D, \,n\geq0\}$ and $\F_\psi=\{\psi_n\in\D, \,n\geq0\}$. Therefore, in
particular, the sets $\D_\varphi$ and $\D_\psi$ of their finite linear combinations are dense in $\H$, and moreover $
\ip{\varphi_i}{\psi_j}=\delta_{i,j}$. We also consider a strictly increasing sequence of non negative numbers:
$0=\epsilon_0<\epsilon_1<\epsilon_2<\cdots$.  On $\D_\varphi$ and $\D_\psi$  we can introduce two operators, $a$ and $b^\dagger$:
$$
\D_\varphi\ni f=\sum_{k=0}^M\,c_k\varphi_k \rightarrow a\,f=\sum_{k=1}^M\,c_k\,\sqrt{\epsilon_{k}}\varphi_{k-1},
$$
and
$$
\D_\psi\ni h=\sum_{k=0}^{M'}\,d_k\psi_k \rightarrow b^\dagger\,h=\sum_{k=1}^{M'}\,d_k\,\sqrt{\epsilon_{k}}\psi_{k-1}.
$$
It is possible to check that $\D_\psi\subseteq D(a^\dagger)$ and $\D_\varphi\subseteq D(b)$ and that, in particular
$$
a^\dagger\psi_k=\sqrt{\epsilon_{k+1}}\,\psi_{k+1},\qquad b\varphi_k=\sqrt{\epsilon_{k+1}}\,\varphi_{k+1},
$$
so that these operators act as {\em generalized rising operators} on two different bases. Analogously, $a$ and $b^\dagger$ act as {\em
generalized lowering operators}, as we can deduce from the formulas above, which give, in particular,
$$
a\varphi_k=\sqrt{\epsilon_{k}}\,\varphi_{k-1},\qquad b^\dagger\psi_k=\sqrt{\epsilon_{k}}\,\psi_{k-1},
$$
if $k\geq1$, or zero if $k=0$. Taking now $f\in \D_\varphi$ and $h\in\D_\psi$ as above, we conclude that \be
\ip{b\,f}{a^\dagger\,h}-\ip{a\,f}{b^\dagger\,h}=\sum_{l=0}^{\min(M,M')}\left(\epsilon_{l+1}-\epsilon_l\right)\,{c_l}\,\overline{d_l}.
\label{f1}\en Let us now introduce an operator $X$ satisfying $X\psi_k=\left(\epsilon_{k+1}-\epsilon_k\right)\psi_k$. This can be formally
written as
\begin{equation}\label{eqn_defnX}
X=\sum_{l=0}^{\infty}\left(\epsilon_{l+1}-\epsilon_l\right)\,\psi_l \otimes \overline{ \varphi_l},
\end{equation}
where $ (\psi_l \otimes \overline{ \varphi_l}) \xi= \ip{\xi}{\varphi_l}\psi_l$. Hence formula (\ref{f1}) can be re-written as
\begin{equation}\label{eqn_formalccr}
\ip{b\,f}{a^\dagger\,h}-\ip{a\,f}{b^\dagger\,h}=\ip{f}{X\,h}.
\end{equation}
Incidentally we observe that in the {\em linear regime}, i.e. when $\epsilon_l=l$, we recover (CR2).

{ In order to make meaningful the above formula \eqref{eqn_formalccr} and  proceed with our analysis, we need a better knowledge of operators
of the form
\begin{equation} \label{X_form}X=\sum_{k=0}^\infty \alpha_k (\psi_k\otimes \overline{ \varphi_k})\end{equation}
with $\{\varphi_n\}$  and $\{\psi_n\}$ two biorthogonal bases and $\alpha_k\geq0$, as above.

To simplify notations, we put $R_k=\psi_k\otimes \overline{ \varphi_k}$. This family of rank one operators enjoys the following easy
properties:
\begin{itemize}
\item[(i)] $\|R_k\|\leq \|\varphi_k\|\, \|\psi_k\|$;
\item[(ii)] $R_k^*= \varphi_k \otimes \overline{\psi_k}$;
\item[(iii)] $R_k^2 =R_k$ and $R_k R_m=0$ if $m\neq k$;
\end{itemize}
In particular, (iii) implies that $R_k$ is a {nonselfadjoint projection} (unless $\varphi_k=\psi_k$). Moreover
$$ \xi= \sum_{k=0}^\infty R_k \xi, \quad \forall \xi \in \D.$$

We notice that the series in the right-hand side converges because of the properties of biorthogonal bases, see \cite{chri,you}. This equality
implies that $\{R_k\}$ is a resolution of the identity.}

{\begin{lemma} \label{thelemma} Let $\{\varphi_n\},\, \{\psi_n\}$ be two biorthogonal bases in $\D$ and let
\begin{equation}\label{eqn_defnX2} X=\sum_{k=0}^\infty \alpha_k ( \psi_k\otimes \overline{ \varphi_k}),\end{equation} with $\{\alpha_n\}$ a sequence of positive real numbers.
Then the following statements hold.
\begin{enumerate}
\item $D(X)\supset \D_\psi$ and $X\psi_k =\alpha_k \psi_k$, for every $k\in {\mb N}$.
\item $\D \subset D(X)$ if, and only if, for every $\xi \in \D$
$$ \lim_{n\to \infty} \left\| \left( \sum_{k=n+1}^{n+p} \alpha_k \psi_k \otimes \overline{ \varphi_k}\right)\xi\right\| =
\lim_{n\to \infty} \left\|  \sum_{k=n+1}^{n+p} \alpha_k\ip{\xi}{\varphi_k}\psi_k \right\|=0,$$
for every $p\in {\mb N}_0$.
\item $\D_\varphi \subset D(X)$ if, and only if, for every $l \in {\mb N}$
$$ \lim_{n\rightarrow\infty}\left\|\sum_{k=n+1}^{n+p}\,\alpha_k\ip{\varphi_l}{\varphi_k}\psi_k\right\|=0,$$
for every $p\in {\mb N}_0$.

\item If $\D \subset D(X)$, $X$ has an adjoint $X^*$ and
 $$X^*=\sum_{n=0}^\infty\alpha_n\,(\varphi_n\otimes\overline{\psi_n}).$$

\end{enumerate}
\end{lemma}
\begin{proof} {(1) is obvious and (2), (3) are nothing but the Cauchy convergence conditions, that we have written explicitly to put in evidence the difference with the case of a single orthonormal basis.
Thus, we prove only (4)}.
First, it is easy to check that
$\sum_{k=0}^\infty\alpha_k\,(\varphi_k\otimes\overline{\psi_k})\subset X^*$.

Conversely, let $\eta$ be an arbitrary element of $D(X^*)$. Then there exists $\zeta\in\H$ such that $$\ip{X\xi}{\eta}=\ip{\xi}{\zeta}, \quad
\forall \xi\in D(X).$$ Since $R_n\xi \in D(X)$ and $XR_n\xi=\alpha_n\,R_n\xi$, for every $\xi \in D(X)$ and $n\in \Bbb N$, we have

{ $$\ip{XR_n\xi}{\eta}= \ip{R_n\xi}{\zeta} =\ip{(\psi_n\otimes\overline{\varphi_n})\xi }{\zeta}=
\ip{\xi}{(\varphi_n\otimes\overline{\psi_n})\zeta}$$}

On the other hand, {$$ \ip{XR_n\xi}{\eta}= \alpha_n\ip{(\psi_n\otimes\overline{\varphi_n})\xi}{\eta}=
\alpha_n\ip{\xi}{(\varphi_n\otimes\overline{\Psi_n})\eta}.
$$
Hence, $(\varphi_n\otimes\overline{\psi_n})\zeta=\alpha_n(\varphi_n\otimes\overline{\psi_n})\eta$, so that
\begin{equation}\label{eqn_finitesums}\sum_{n=0}^N \alpha_n(\varphi_n\otimes\overline{\psi_n})\eta = \sum_{n=0}^N (\varphi_n\otimes\overline{\psi_n})\zeta. \end{equation}
Moreover, since { $\sum_{n=0}^\infty (\varphi_n \otimes \overline{\psi_n})\zeta= \sum_{n=0}^\infty R_n^*\zeta =\zeta$}, by a limiting procedure in \eqref{eqn_finitesums}, it follows that $\eta$ belongs to the domain of the operator $ \sum_{n=0}^\infty\alpha_n (\varphi_n \otimes \overline{\psi_n})$ and
$\sum_{n=0}^\infty\alpha_n (\varphi_n \otimes \overline{\psi_n})\eta=\zeta.$}

\end{proof}

}
\medskip
Of course, these conditions are clearly satisfied when $\F_\varphi$ and $\F_\psi$ collapse into a single orthonormal set.

\medskip

We are now ready to introduce the following definition:

\bedefi\label{def1}  Let $S, T \in \LDH$ and $\{\varphi_n\}$  and $\{\psi_n\}$ two biorthogonal bases of $\H$, contained in $\D$. We say that
$S$ and $T$ satisfy the nonlinear CR.2 if, for all $\xi$ and $\eta$ in $\D$, \be
\ip{T\,\xi}{S^\dagger\,\eta}-\ip{S\,\xi}{T^\dagger\,\eta}=\ip{\xi}{X\,\eta}, \label{31}\en {where $X$ is an operator of the form
\eqref{eqn_defnX2} with $\D\subset D(X)$}. \findefi

\berem \label{rem_ONbasis}Let $\{\chi_n\}$ be an orthonormal basis in $\D$ and $G$ a symmetric bounded operator with bounded inverse $G^{-1}$.
Suppose that $G\D=\D$. Then, if we put $\vp_k :=G\chi_k$  and $\psi_k:= G^{-1}\chi_k$, we obtain two biorthogonal bases of $\H$, still
belonging to $\D$. Under these assumptions, we get
$$(\psi_k \otimes \overline{\varphi_k})\xi = \ip{\xi}{\varphi_k}\psi_k =  \ip{\xi}{G\chi_k}G^{-1}\chi_k = \ip{G\xi}{\chi_k}G^{-1}\chi_k.$$
Hence $\psi_k \otimes \overline{\varphi_k}= G^{-1}(\chi_k \otimes \overline{\chi_k})G$. Thus if $S,T$ satisfy the non linear CR.2 with $X$ as
in \eqref{X_form}, the operators $K:=G^{-1}TG$ and $H:=G^{-1}SG$ satisfy
$$\ip{K\,\xi}{H^\dagger\,\eta}-\ip{H\,\xi}{K^\dagger\,\eta}=\ip{\xi}{Y\,\eta} , \quad \forall \xi, \eta \in \D,
$$
where \begin{equation} \label{eqn_Y} Y= \sum_{k=0}^\infty \alpha_k (\chi_k\otimes \overline{ \chi_k}).\end{equation} {  Therefore,
$X=G^{-1}YG$, and $D(X)\supset \D$ if, and only if,  $$\sum_{k=0}^\infty\alpha_k^2|\ip{G\xi}{\chi_k}|^2<\infty, \quad\forall \xi\in\D.$$} It is natural to consider, as a first step, this simpler situation. \enrem

The operator $Y$, defined in \eqref{eqn_Y} is bounded if, and only if, $\{\alpha_k\}\in l^\infty(\Bbb N)$. Indeed, suppose first that $Y$ is
bounded. { Hence, for some $M>0$, $\|Yf\|\leq M\|f\|$, for every $f\in\H$.} Then, for all $k\in {\mb N}$, ${\alpha_k}\|\chi_k\|=\|Y\chi_k\|\leq
M\|\chi_k\|$. Therefore $|\alpha_k|\leq M$, for all $k$. Viceversa, let us assume that ${0\leq\alpha_k}\leq M$, for all $k\in {\mb N}$. Then, using the
orthogonality of the $\chi_k$'s and the Parceval equality,
$$
\|Yf\|^2=\sum_{k=0}^\infty {\alpha_k^2}|\ip{\chi_k}{f}|^2\leq M^2\|f\|^2.
$$
{So $Y$ is bounded and $\|Y\|\leq M$}.

The spectrum $\sigma(Y)$ is also easily determined: in fact, { $\sigma(Y)=\sigma_p(Y)=\overline{\{\alpha_k; k \in {\mb N}\}}$}, where $\sigma_p(Y)$ denotes,
as usual, the point spectrum of $Y$. We remark that in finite dimensional spaces every family of projections whose sum is the identity operator
is similar to a family of orthogonal projections; so that the situation discussed above is the more general possible. For the infinite
dimensional case, an analogous statement was obtained by Mackey \cite[Theorem 55]{mackey}: every nonselfadjoint resolution of the identity
(i.e. a spectral measure on the Borel set of the plane or of the real line) is similar to a selfadjoint resolution of the identity.

{ The extension to $X$ of the results outlined in Remark \ref{rem_ONbasis} is, under suitable assumptions, quite straightforward.

\medskip

We have

\begin{prop} Let  $\F_\psi=\{\psi_k\}$ and $\F_\varphi=\{\varphi_k\}$ be  biorthogonal Riesz bases for $\H$ and let $$X=\sum_{k=0}^\infty \alpha_k (\psi_k\otimes \overline{ \varphi_k}), \quad \alpha_k \in {\mb R}^+.$$ Then the following statements hold.
\begin{itemize}
\item[(i)] $X$ is bounded if and only if $\{\alpha_k\}\in l^\infty(\Bbb N)$.  
\item[(ii)] For every $k\in {\mb N}$, $\psi_k\in D(X)$, $\varphi_k\in D(X^*)$, and
$$
X\psi_k=\alpha_k\psi_k,\qquad X^*\varphi_k=\alpha_k\,\varphi_k.
$$
\item[(iii)] $\sigma(X)=\overline{\{\alpha_k; k \in {\mb N}\}}.$
\end{itemize}

\end{prop}
\begin{proof}
(i):  By (1) of Lemma \ref{thelemma} it follows immediately that if $X$ is bounded, then  $\{\alpha_k\}$ is bounded too and
$\sup_{k\in {\mb N}}{\alpha_k}\leq \|X\|$. On the other hand, let $\{\alpha_k\}$ be bounded. Let now $S$ be the bounded, self-adjoint, operator, with
bounded inverse, satisfying $\varphi_n=S\,e_n$, $\Psi_n=S^{-1}e_n$, $n=0,1,2,\ldots$, where $\{e_n\}$ is an orthonormal basis of $\H$. Then
$$ \left\|X \sum_k \beta_k\psi_k\right\|\leq \|S\|\|S^{-1}\|\sup_{k\in {\mb N}}{\alpha_k}\left\|\sum_k \beta_k\psi_k\right\|.$$
The density of $\D_\psi$ in $\H$ implies the statement.

(ii): The proof of this statement easily follows from the definition of $X$ and from the biorthogonality of $\F_\psi$ and $\F_\varphi$. \\
(iii): By (ii), the numbers $\{\alpha_k\}$ are eigenvalues of $X$. The spectrum $\sigma(X)$ consists of the closure $\overline{\{\alpha_k; k
\in {\mb N}\}}$ of { the set of} these eigenvalues. Indeed, {let $\lambda \in {\mb C}\setminus\overline{\{\alpha_k; k\in {\mb N}\}}$ and define} $$Z= \sum_{k=0}^\infty
\frac{1}{\alpha_k-\lambda}R_k.$$ Then,
\begin{eqnarray*} (X-\lambda\Id)Z\xi &=& (X-\lambda\Id)\sum_{k=0}^\infty \frac{1}{\alpha_k-\lambda}R_k\xi\\
&=&\sum_{m=0}^\infty (\alpha_m-\lambda)\sum_{k=0}^\infty \frac{1}{\alpha_k-\lambda}R_mR_k\xi\\
&=& \sum_{k=0}^\infty R_k\xi =\xi.
\end{eqnarray*}
{Since $\lambda$ is not a limit point of the set $\{\alpha_k; k \in {\mb N}\}$, the sequence $\{(\alpha_k-\lambda)^{-1}\}$ is bounded and thus, by (i) $Z$ is bounded and has a continuous extension to $\H$. Hence, $X-\lambda\Id$ has a bounded inverse. On the other hand, every $\lambda \in
\overline{\{\alpha_k; k \in {\mb N}\}} \setminus \{\alpha_k; k \in {\mb N}\}$, if any, belongs to the continuous spectrum $\sigma_c (X)$ of $X$.}
\end{proof}

A slightly weaker result can be obtained if we require, in $X$, that the two sets $\F_\psi$ and $\F_\varphi$ are biorthogonal but not
necessarily complete in $\H$. With a similar argument as before we conclude that, if $X$ is bounded, then $\{\alpha_k\}\in l^\infty(\Bbb N)$, {
but the vice-versa does not hold, in general}. The {\em Riesz-like} nature of the sets is not important, here.

\subsection{Consequences of Definition \ref{def1}}

{ {\em Raising} and {\em lowering} operators play a crucial role in connections with bosons and their generalizations. Hence it is natural to consider this aspect in connection with Definition \ref{def1}.

Let $S,T \in \LDH$ satisfy \eqref{31} and assume that, for some $k\in {\mb N}$,  $T\varphi_k, \, S\varphi_k\in\D$, then,
 $$\ip{(ST\varphi_k-TS\varphi_k)}{\eta}=\alpha_k\ip{\varphi_k}{\eta}, \quad \forall\eta\in\D,$$ so that the
$$
ST\varphi_k- TS \varphi_k=\alpha_k\varphi_k.
$$
}

Analogously, if $T\ad\psi_k, \,S\ad\psi_k\in \D$, then $T\ad S\ad \psi_k - S\ad
T\ad \psi_k =\alpha_k\psi_k$.

\medskip
A first simple {\em raising and lowering} property can be stated as follows: suppose
$S\varphi_0=0$. Then  $T\varphi_0\neq0$ and $T\varphi_0$ is an eigenvector of the (formal) operator $N_l:=TS\das$ with eigenvalue $\alpha_0$.
Moreover, if $T\varphi_0\in\D$ , then $\varphi_0$ is eigenvector of the (formal) operator $N_r:=ST$ with the same eigenvalue, $\alpha_0$. For
convenience we also put $N_l^\#:=S\ad T^* $ and $N_r^\#:=T\ad S^*$. Analogously, if $T\ad\psi_0=0$, then  $S\ad\psi_0\neq0$ and $S\ad\psi_0$ is
an eigenvector of {$S\ad T^*$} with eigenvalue $\alpha_0$. Moreover, if $S\ad\psi_0\in\D$, then $\psi_0$ is eigenvector of {$T\ad S^*$}  with
the same eigenvalue, $\alpha_0$.


\begin{prop}\label{prop1}
Suppose that $S\varphi_0=0$ and that all the $\alpha_k$'s are different. Then the following statements are equivalent.
\begin{enumerate}
\item $X^*(T\varphi_n)=\alpha_{n+1}(T\varphi_n)$, for all $n\in {\mb N}$.
\item For every $n \in {\mb N}$, there exists $\gamma_n \in {\mb C}$ such that $T\varphi_n=\gamma_n\,\varphi_{n+1}$.
\item For every $n \in {\mb N}$, there exists ${\beta_{n}} \in {\mb C}$ such that $T\ad\psi_n={\beta_{n}}\,\psi_{n-1}$, where $\beta_{0}:=0$.
\end{enumerate}
In this case $\beta_n=\overline{\gamma_{n-1}}$ for all $n\geq1$.
\end{prop}
\begin{proof}
{We may suppose that the $\gamma_n$'s are not zero, since otherwise  (1), (2) and (3)
are trivially equivalent. }

(1)$\Rightarrow$(2): By taking the scalar product of both sides of the equality in (1) with a generic $\psi_l$, we obtain
$$
(\alpha_{n+1}-\alpha_l)\ip{T\varphi_n}{\psi_l}=0,
$$
for all possible $l$. Since the $\alpha_k$'s are different, if $l\neq n+1$, the vector $T\varphi_n$ must be orthogonal to $\psi_l$, $\forall
l\in {\mb N}\setminus\{n+1\}$. Hence, due to the uniqueness of the biorthogonal basis \cite{chri}, $T\varphi_n$ is necessarily proportional to
$\varphi_{n+1}$. Then (2) follows.

(2)$\Rightarrow$(3): { Using (2) and the biorthogonality condition $\ip{\varphi_n}{\psi_m}=\delta_{n,m}$, we get
$\ip{T\varphi_n}{\psi_l}=\gamma_{l-1}\delta_{n,l-1}$. On the other hand, $\ip{T\varphi_n}{\psi_l}=\ip{\varphi_n}{T\ad\psi_l}$. Thus, for $l\geq1$,
$$\ip{\varphi_n}{\left(T\ad\psi_l-\overline{\gamma_{l-1}}\psi_{l-1}\right)}=0, \quad \forall n \in {\mb N}.$$} Then (3) follows from the
completeness of $\F_\varphi$ with ${\beta_{n}}=\overline{\gamma_{n-1}}$. Notice also that, if $l=0$, (3) is trivially true because of the
assumption $T\ad\psi_0=0$ and since $\beta_0=0$.

(3)$\Rightarrow$(1): By (3) we get easily the equality $T\varphi_n=\overline{\beta_{n+1}}\varphi_{n+1}$. Therefore, { by (4) of Lemma \ref{thelemma}}, $T\varphi_n\in D(X^*)$ and

$$X^*(T\varphi_n)=\left(\sum_{k=0}^\infty\alpha_k\,(\varphi_k\otimes\overline{\psi_k})\right)\overline{\beta_{n+1}}\varphi_{n+1}=\alpha_{n+1}(T\varphi_n).$$

\end{proof}

Analogous results can be proved for the operator $S$. Indeed, we have

\begin{prop}\label{prop2}
Suppose that $T\ad\psi_0=0$ and that all the $\alpha_k$'s are different. Then the following statements are all equivalent.
\begin{enumerate}
\item $X(S\ad\psi_n)=\alpha_{n+1}(S\ad\psi_n)$ for all $n\in {\mb N}$.
\item For every $n \in {\mb N}$, there exists $\tilde\gamma_n \in {\mb C}$ such that $S\ad\psi_n=\tilde\gamma_n\,\psi_{n+1}$.
\item For every $n \in {\mb N}$, there exists $\tilde{\beta_{n}} \in {\mb C}$ such that  $S\varphi_n=\tilde{\beta_{n}}\,\varphi_{n-1}$, where $\tilde\beta_0:=0$.
\end{enumerate}
In this case $\tilde{\beta_{n}}=\overline{\tilde\gamma_{n-1}}$, for every $n \geq1.$
\end{prop}
These two propositions, together, have interesting consequences:

\begin{cor}\label{cor47}
Suppose that $S\varphi_0=0$, $T\ad\psi_0=0$,  and that all the $\alpha_k$'s are different.
Suppose also that $X^*(T\varphi_n)=\alpha_{n+1}(T\varphi_n)$ and $X(S\ad\psi_n)=\alpha_{n+1}(S\ad\psi_n)$,
for all $n\geq0$. Then , for all $n\geq0$, \\

\begin{enumerate}
\item $N_l\varphi_n=\gamma_{n-1}\overline{\tilde\gamma_{n-1}}\varphi_n$ and $N_r\varphi_n=\gamma_{n}\overline{\tilde\gamma_{n}}\varphi_n$;
\item $N_l^\#\psi_n=\overline{\gamma_{n-1}}\tilde\gamma_{n-1}\psi_n$ and $N_r^\#\psi_n=\overline{\gamma_{n}}\tilde\gamma_{n}\psi_n$;
\item $\alpha_n=\gamma_{n}\overline{\tilde\gamma_{n}}-\gamma_{n-1}\overline{\tilde\gamma_{n-1}}$.
\end{enumerate}
\end{cor}
\begin{proof}
We will only prove here statement (3), since the others are easy consequences of the previous Propositions. First of all, from Propositions
\ref{prop1} and \ref{prop2}, we get $\ip{T\varphi_n}{S^\dagger\psi_m}=\gamma_n\overline{\tilde\gamma_n}\delta_{n,m}$. Moreover {$$
\ip{T\varphi_n}{S^\dagger\psi_m}=\ip{\varphi_n}{X\psi_m}+\ip{S\varphi_n}{T\ad\psi_m}=(\alpha_n+\gamma_{n-1}\overline{\tilde\gamma_{n-1}})\delta_{n,m},
$$
where we have used $X^*\varphi_n=\alpha_n\varphi_n$}. Hence (3) follows.

\end{proof}

\vspace{2mm}

\begin{rem} Since $\alpha_n\geq0$ for all $n$, then { $\gamma_{n}\overline{\tilde\gamma_{n}}- \gamma_{n-1}\overline{\tilde\gamma_{n-1}}\geq 0$}, for every $n\geq0$. Also,
since $\gamma_{-1}=\tilde\gamma_{-1}=0$, we  find that $\alpha_0=\gamma_{0}\overline{\tilde\gamma_{0}}$, and that, for all $n\geq1$, \be
\alpha_n=\gamma_{n}\overline{\tilde\gamma_{n}}-\sum_{k=0}^{n-1}\alpha_k, \label{32}\en which provides a relation between the $\alpha_k$'s and
the $\gamma_k$'s, $\tilde\gamma_k$'s introduced previously.
\end{rem}

{Propositions \ref{prop1} and \ref{prop2} are  related, in a sense, by the following}

\begin{prop}
Suppose that  $S\varphi_0=0$, $T^\dagger\psi_0=0$,  and that all the $\alpha_k$'s are different. {Then $T\varphi_n=\gamma_n\varphi_{n+1}$, for
every $n \in {\mb N}$, if, and only if, $S\varphi_n=\overline{\tilde\gamma_{n-1}}\varphi_{n-1}$, for every $n\geq 1$.}
\end{prop}
\begin{proof}
We use induction on $n$. Let us first suppose that $T\varphi_n=\gamma_n\varphi_{n+1}$, for all $n\geq0$. Then, in particular,
$\varphi_1=\frac{1}{\gamma_0}\,T\varphi_0$. Now, taking $f\in\H$,
$$
\ip{S\varphi_1}{f}=\frac{1}{\gamma_0}\,\ip{ST\varphi_0}{f}=\frac{\alpha_0}{\gamma_0}\,\ip{\varphi_0}{f}=\ip{\overline{\tilde\gamma_0}\varphi_0}{f},
$$
so that, because of the arbitrariness of $f$, $S\varphi_1=\overline{\tilde\gamma_0}\varphi_0$. 

Let us now assume that $S\varphi_n=\overline{\tilde\gamma_{n-1}}\varphi_{n-1}$. We want to check that
$S\varphi_{n+1}=\overline{\tilde\gamma_{n}}\varphi_{n}$. In fact, from the hypothesis, we deduce that
$\varphi_{n+1}=\frac{1}{\gamma_n}\,T\varphi_n$. Therefore, {$$
S\varphi_{n+1}=\frac{1}{\gamma_n}\,S\,T\varphi_n=\frac{1}{\gamma_n}\,\left(X+T\,S\right)\varphi_n.
$$}
Now, recall that $X\varphi_n=\alpha_n\varphi_n$. Moreover, using the induction hypothesis, we have
$T\,S\,\varphi_n=\overline{\tilde\gamma_{n-1}}\,T\,\varphi_{n-1}=\overline{\tilde\gamma_{n-1}}\,\gamma_{n-1}\,\varphi_{n}$.
Hence, by (3) of Corollary \ref{cor47}, we conclude that $S\varphi_{n+1}=\overline{\tilde\gamma_{n}}\varphi_{n}$.

The inverse implication can be proved in a similar way.

\end{proof}

\vspace{2mm}

\begin{rem} A similar result can be proved for $T^\dagger$ and $S^\dagger$ in connection with $\F_\Psi$.

\end{rem}

Adopting the standard notation for nonlinear coherent states we call $\gamma_n!=\gamma_0\,\gamma_1\cdots\gamma_n$ and
$\tilde\gamma_n!=\tilde\gamma_0\,\tilde\gamma_1\cdots\tilde\gamma_n$. Iterating the formulas in the previous Propositions it is easy to find
that
$$
\varphi_n=\frac{1}{\gamma_n!}\,T^n\,\varphi_0,\qquad \psi_n=\frac{1}{\tilde\gamma_n!}\,(S^\dagger)^n\,\psi_0.
$$
All the above formulas provide a rather natural interpretation of $T$, $S$ and their adjoints as lowering and rising operators with respect to
two different bases, as usually done in pseudo-hermitian quantum mechanics (see \cite{nlpb} and references therein). {In that framework, moreover, the existence of of some {\em intertwining
operators} plays an important role. In a weaker sense they can be introduced also here.

Indeed, we can define two operators $S_\varphi$ and $S_\psi$, via their action on two generic vectors, $f\in D(S_\varphi)$ and
$g\in D(S_\psi)$, as follows
$$
S_\varphi\,f=\sum_{k=0}^\infty\ip{f}{\varphi_k}\,\varphi_k, \qquad S_\psi\,g=\sum_{k=0}^\infty\ip{g}{\psi_k}\,\psi_k.
$$
These operators are densely defined (since $S_\varphi\,\psi_k=\varphi_k$ and $S_\psi\,\varphi_k=\psi_k$,
$\forall\,k\in {\mb N}$) and positive. Moreover,  $S_\varphi\,S_\psi\varphi_k=\varphi_k$ and $S_\psi\,S_\varphi\psi_k=\psi_k$, for
all $k$. If, in addition, they are bounded, then they are inverses of each other; i.e., $S_\psi=S_\varphi^{-1}$, as it happens when $\F_\varphi$
and $\F_\psi$ are Riesz bases. In our case the following {\em weak intertwining
relations} hold:}
$$
N_l\,S_\varphi\,\psi_k=S_\varphi\,N_l^\#\psi_k,\qquad N_r\,S_\varphi\,\psi_k=S_\varphi\,N_r^\#\psi_k, \quad \forall k \in {\mb N}
$$
and
$$
N_l^\#\,S_\psi\,\varphi_k=S_\psi\,N_l\varphi_k,\qquad N_r^\#\,S_\psi\,\varphi_k=S_\psi\,N_r\varphi_k, \quad \forall k \in {\mb N}.
$$
 The existence of these relations is not surprising, since it is clearly related to the fact that, for instance, $N_l$ and $N_l^\#$
have equal eigenvalues.

\subsection{Connections with the linear case}

We end the paper by discussing some relations between the present situation, i.e. the nonlinear case, with the one discussed in
\cite{bit_jmp2011} and in  Sections 2 and 3. In particular, we will show that our previous results could be considered as special cases of the
present settings.

The starting point is Definition \ref{def1}. In order to recover here similar results to those obtained in \cite{bit_jmp2011}, we assume
that a non zero vector $\Phi$ does exist in $\D$ which is annihilated by $S$, $S\Phi=0$, and such that $T^k\Phi$ exists and is an eigenvector
of $X^*$: $X^*(T^k\Phi)=\mu_k(T^k\Phi)$. Under these assumptions we can check that $T^k\Phi$ is an eigenvector of ${S^\dagger}^*\,T$ with
eigenvalue $M_k:=\mu_0+\mu_1+\cdots+\mu_k$, and that $T^{k+1}\Phi$ is an eigenvector of $T\,{S^\dagger}^*$ with the same eigenvalue, $M_k$. The
proof, which can be given by induction on $k$, is easy and will not be given here. {It is worth remarking} that these
assumptions are satisfied whenever we are in the situation briefly considered at the beginning of Section 4. In fact, in this case, it is
enough to take $\Phi=\varphi_0$ and $T=b$. Hence, since $T^k\Phi=b^k\varphi_0=\sqrt{\epsilon_k!}\varphi_k$, using the explicit expression for
$X$ we conclude that $X^*(T^k\Phi)=\alpha_k(T^k\Phi)$.

Secondly, we give  the following result

{\begin{prop} Let us assume that $S,T\in\LDH$ satisfy Definition \ref{def1}. Let $n \in {\mb N}$, $n\geq 1$, and   $\xi$  a vector in $\D$ such
that  $T^k\xi\in\D$ for $k\leq n$. Then,
\begin{itemize}

\item[(i)] $S\xi\in D(({T^\dagger}^*)^k)$, for $k\leq n$

\item[(ii)] $X^*T^l\xi\in D(({T^\dagger}^*)^m)$, for all $l, m$ such that $l+m=k-1$, $k\leq n$
\item[(iii)] the following equality holds
\be ST^k\xi-({T\das})^kS\xi=\sum_{l=0}^{k-1}({T^\dagger}^*)^{k-1-l}\,X^*\,T^l\xi, \quad \forall k\leq n. \label{33}\en
\end{itemize}
\end{prop}}
\begin{proof}
The proof is given by induction on $n$.

{For $n=1$ the statements follow immediately from Definition \ref{def1}.

Let us assume that (i), (ii), (iii) hold for $n$ and let $k\leq n+1$. If $k\leq n$, then the statements follow by the induction assumptions.
Thus we need to prove it only for $k=n+1$. Assume then that also $T^{n+1}\xi\in \D$. Then the vector $\xi'=T\xi \in \D$ satisfies $T^k\xi'\in\D$
for $k\leq n$. Thus the induction assumptions apply to $\xi'$. Therefore $ST\xi \in D(({T^\dagger}^*)^k)$, for $k\leq n$; $X^*T^{l+1}\xi\in
D(({T^\dagger}^*)^m)$, for all $l, m$ such that $l+m=k-1$, $k\leq n$. We prove that equation (\ref{33}) holds for $n+1$. Indeed, we have}
$$
ST^{n+1}\xi=ST^n(T\xi)=({T^\dagger}^*)^nS(T\xi)+\sum_{l=0}^{n-1}({T^\dagger}^*)^{n-1-l}X^*T^l(T\xi)=
$$
$$
=({T^\dagger}^*)^n\left(X^*\xi+{T^\dagger}^*\,S\,\xi\right)+\sum_{l=0}^{n-1}({T^\dagger}^*)^{n-1-l}X^*T^{l+1}\xi,
$$
from which formula (\ref{33}) for $n+1$ follows.
\end{proof}

{\bf Remarks:--} (1) Notice that, if $X=\Id$, i.e. if $\alpha_k=1$ for all $k$ in the definition of $X$, we recover Propositions 3.2 and 3.4 of
\cite{bit_jmp2011}. In particular, (\ref{33}) becomes $ST^k\xi-(T\das)^kS\xi=k\,T^{k-1}\xi$.

(2) If, rather than this, we simply assume that $[X^*,T]\xi=0$ for all $\xi\in\D$, and that $T^l(X^*\xi)\in\D$ for all $l$, the right-hand side
of (\ref{33}) becomes $kX^*T^{k-1}\xi$, which again, returns the previous result when $X=\Id$.

\section*{Acknowledgements}

This work was partially supported by the Japan Private School Promotion Foundation and partially by CORI, Universit\`a di Palermo. F.B. and
C.T. acknowledge the warm hospitality of the Department of Applied Mathematics of the Fukuoka University. F.B. also acknowledges partial
financial support by GNFM, Italy.

The authors wish to thank the {referee} for his very many useful comments.

\end{document}